\documentclass[11pt,a4paper]{article}
\usepackage[margin=1in]{geometry}
\usepackage{xspace}
\usepackage{amsmath,amssymb,amsthm}
\usepackage{graphics}
\usepackage{graphicx}
\usepackage{hyperref}
\usepackage{tikz}
\usepackage{paralist}
\usepackage{wrapfig}
\usepackage{bibunits}

\newtheorem{theorem}{Theorem}

\newtheorem{corollary}{Corollary}

\defaultbibliography{abbrv}
\defaultbibliography{regular_augmentation_shortAndProofs}

\newtheorem {defn}{Definition}

\theoremstyle{definition}
\newtheorem{exmp}{Example}

\def\comment#1{}

\def\withcomments{
  \newcounter{mycommentcounter}
  \def\comment##1{\refstepcounter{mycommentcounter}%
   \ifhmode%
    \unskip%
    {\dimen1=\baselineskip \divide\dimen1 by 2 %
      \raise\dimen1\llap{\tiny -\themycommentcounter-}}\fi%
    \marginpar{\renewcommand{\baselinestretch}{0.8}%
      \footnotesize [\themycommentcounter]: \raggedright ##1}}
  }
\withcomments

\title{Communication-Efficient and Exact Clustering \\Distributed Streaming Data}
\author{Dang-Hoan Tran}%
\date{Ilmenau University of Technology\\\texttt{dang-hoan.tran@tu-ilmenau.de}}%

\begin{document}

\maketitle

\begin{abstract}
  A widely used approach to clustering a single data
stream is the two-phased approach in which the online phase
creates and maintains micro-clusters while the off-line phase
generates the macro-clustering from the micro-clusters. We use
this approach to propose a distributed framework for clustering
streaming data. Our proposed framework consists of fundamental
processes: one coordinator-site process and many remote-site
processes. Remote-site processes can directly communicate with the coordinator-process but cannot communicate the other remote site processes.
Every remote-site process generates and maintains
micro-clusters that represent cluster information summary,
from its local data stream. Remote sites send the local micro-clusterings
to the coordinator by the serialization technique, or
the coordinator invokes the remote methods in order to get the
local micro-clusterings from the remote sites. After the coordinator receives all the local
micro-clusterings from the remote sites, it generates the global
clustering by the macro-clustering method.

Our theoretical and empirical results show that, the global clustering
generated by our distributed framework is similar to the
clustering generated by the underlying centralized algorithm on
the same data set. By using the local micro-clustering approach,
our framework achieves high scalability, and communication-efficiency.
\end{abstract}

\section{Introduction}
\label{sec:introduction}
Mining distributed data streams has been increasingly received great attention \cite{zaki2002introduction,sun2006distributed,halkidi2011online}. Data  in nowadays data analysis applications is too big to gather and analyze in a single location or data is generated by the distributed sources. Design and development of data stream mining algorithms in high-performance and distributed environments thus should meet system scalability and interactivity. This paper studies the problem of clustering distributed streaming data.

Clustering is considered as an unsupervised learning method which classifies the objects into groups of similar objects \cite{jain1999data}. Clustering data stream continuously produces and maintains the clustering structure from the data stream in which the data items continuously arrive in the ordered sequence \cite{guha2003clustering}. There are two types of problems of clustering data streams: (1) Clustering streaming data is to classify the data points that come from a single or multiple data streams into groups of similar data points \cite{aggarwal2003framework,cao2006density,kranen2009self}; (2) Clustering multiple data streams is to classify the data streams into groups of data streams of similar behavior or trend \cite{beringer2006online,Dai2004}.
The data stream processing is a one-time passing, and online data processing model while the systems resources are limited. The basic requirements for clustering data streams \cite{barbara2002requirements} includes a compact representation
of clustering structures, fast, incremental processing of recently incoming
data items, and clear and fast identification of \textquotedblleft{}outlier\textquotedblright{}.
\begin{figure}
    \centering
         \includegraphics[width=0.45\textwidth]{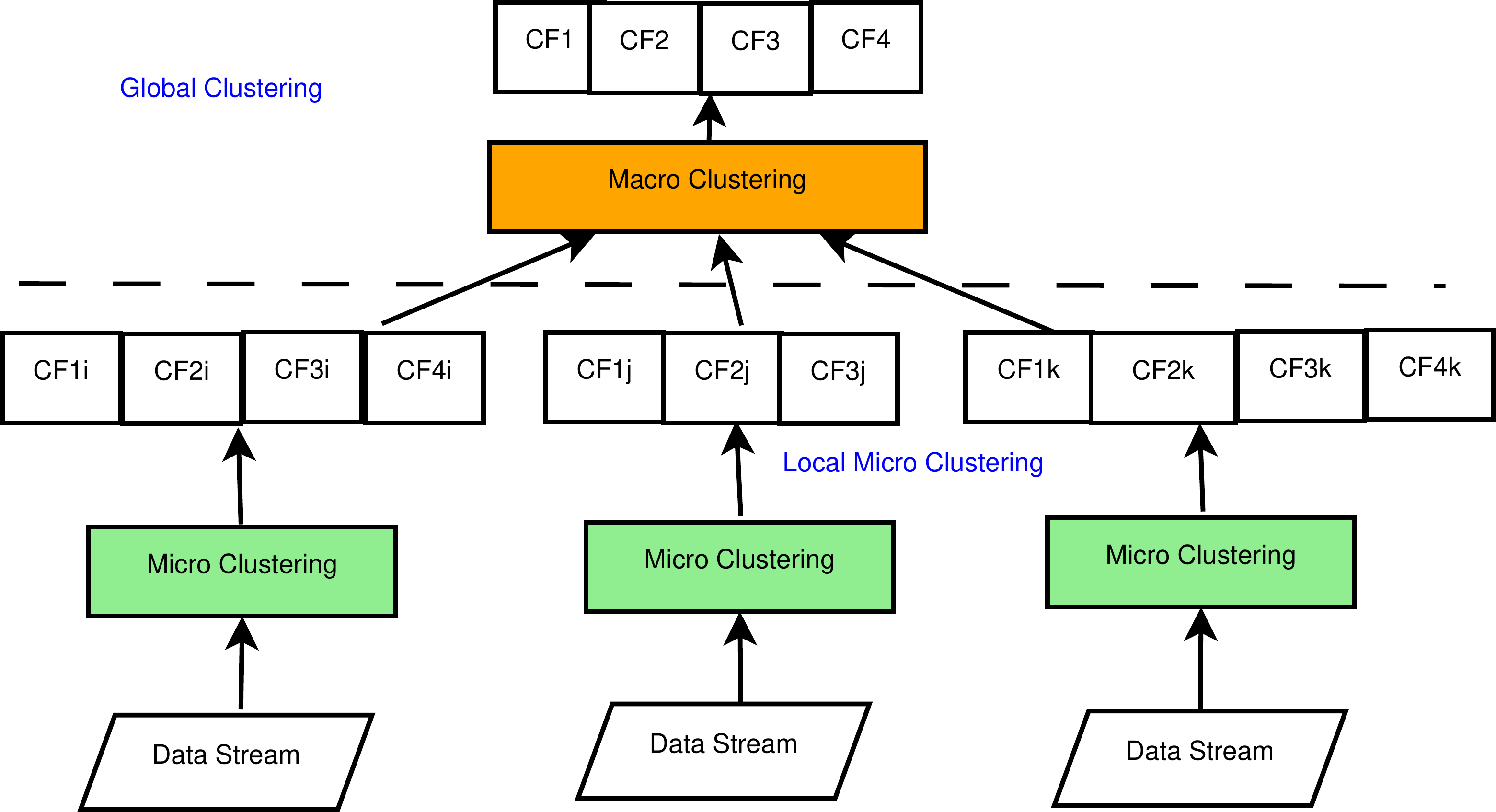}
    \caption{A distributed two-staged framework for clustering streaming data:(1) At the remote sites, micro clusters are created by using micro clustering algorithm; (2) The global clustering can be generated from local micro clusters by using the macro clustering algorithm}
    \label{fig:diststreamarchitecture}
\end{figure}
Clustering distributed streaming data is to classify the data points that come from multiple data streams generated by the distributed sources.
The problem of clustering distributed streaming data can be used to solve many real world applications as follows.
\begin{exmp}
We consider a sensor network that continuously monitors the changes in the environment such as a building. In a naive approach, every sensor continuously sends its stream of observations (temperature, humidity, light) to the base station for processing in order to give a global view of the environment. However, the challenges facing this approach include: (1) the limited communication bandwidth; (2) the overlapping of observations sensed by nearby sensors that generates data redundancy. Clustering distributed streaming sensor data may be a solution to this problem.
\end{exmp}

\begin{exmp}
We consider the problem of discovery of topics that emerge from multiple distributed text streams in Internet. One of the solutions to this problem is to reduce it into the problem of clustering distributed text streams \cite{DBLP:conf/sac/ZimmermannNSSK12}.
\end{exmp}
\section{Problem Formulation}
\label{sec:formalmodel}
We consider a network of one coordinator site and $N$ remote sites. Let $\left\{ S_{1},\,...,S_{N}\right\} $ be the incoming data streams to the  remote clustering modules at $N$ remote sites. A data stream is an infinite sequence of elements $S_{i}=\left\{ \left(X_{1},\, T_{1}\right),..,\left(X_{j},\, T_{j}\right),...\right\} $.
Each element is a pair $\left(X_{j},\, T_{j}\right)$ where $X_{j}$ is a $d$-dimensional vector $X_{j}=\left(x_{1},\, x_{2},...,\,x_{d}\right)$ arriving at the time stamp $T_{j}$. We assume that $S_{i}^{t}$ is a block of $M$ data items that arrives at the remote site $i$ during each update epoch. We also assume that all the data streams arrive at the remote sites with the same data speed, that means at each update epoch $t$, every remote site receives the same $M$ data items. We assume that each site knows its own clustering structure but nothing about the clustering structures at other sites.
Let $\mathcal{A}$ denote an underlying two-phased stream clustering algorithm. Let $Mic\mathcal{A}$ and $Mac\mathcal{A}$ denote the micro-clustering algorithm and the macro-clustering algorithm of $\mathcal{A}$ respectively.

Let $LC^{t}_{i}=\left\{ CF_{i}^{1},..,..,CF_{i}^{j},..,CF^{K_{i}}_{i}\right\}$ be a local micro-clustering generated by a clustering algorithm at remote site $i$ at the update epoch $t$, where $CF_{i}^{j}$ is the $j-th$ micro-cluster, $K_{i}$ is the number of micro-clusters, for $i=1,..,N$.

One of the fundamental properties of distributed computational systems is locality \cite{naor1993can}. A distributed algorithm for clustering streaming data should meet the locality. A local algorithm is defined as one whose resource consumption is independent of the system size. Local algorithms can fall into one of two categories \cite{datta2006distributed}: (1) exact local algorithms are defined as ones that produce the same results as a centralized algorithm; (2) approximate local algorithms are algorithms that produce approximations of the results that a centralized algorithms would produce.  Two attractive properties of local algorithms are scalability and fault tolerance. Recently, Wolff et al. \cite{wolff2008generic} have presented a generic local algorithm for mining streaming data in large distributed systems. By using the local clustering algorithms, we can achieve the scalability of distributed stream clustering algorithms. A distributed framework for clustering streaming data can be robust to network partitions, and node failures. The objective we want is to create and maintain the global clustering which is similar to the clustering created by the centralized stream clustering algorithm. We can achieve this objective by using the serialization technique and the local micro-clustering algorithms.

A distributed streaming data clustering algorithm based on the underlying algorithm $\mathcal{A}$ is called  $Dis\mathcal{A}$. The micro-clustering $Mic\mathcal{A}$ creates and maintains micro-clusters at the remote sites.
The macro-clustering algorithm $Mac\mathcal{A}$ is used to produce meaningful macro-clustering at the coordinator.

\section{Algorithm Description} \label{sec:distclustering}
Our framework is organized as a client/server model in which the coordinator works as a server while the remote sites serve as the clients.
We assume that our distributed stream clustering algorithm is used to answer the continuous queries. Continuous queries can be classified into the
following categories: (1) push-based query processing model
in which remote sites publish/advertise their information to the
coordinator site for further processing; (2) pull-based query
processing model in which the coordinator disseminates the
interests of the user to a set of remote sites. In this case,
the interest may be related to an attribute of the physical
field (query) or an event of interest to be observed (task).
Based on the user interest, the remote sites respond with
the requested information.
Our distributed framework for clustering streaming data includes two fundamental processes: remote-site process and coordinator-site process.
On the basis of type of query, we design the appropriate communication protocols as follows.
\begin{itemize}
\item \emph{Pushed-based query:} If the clustering query is push-based query, the remote sites send the local micro-clusterings to the coordinator by the serialization technique. As the coordinator process is a passive process, it continuously listens for connection requests from the remote sites and receives the local micro clusterings that are sent by the remote sites.
\item \emph{Pull-based query:} If the clustering query is pull-based query, the coordinator invoke the remote methods on the local data streams from the remote sites in order to get the local micro-clusterings.
\end{itemize}
In order for the coordinator to communicate with many remote sites concurrently, the global clustering process is organized as a multi-threading process.

\subsection{Data Structure}
Our algorithm uses the same data structure \emph{Clustering} to represent both local clustering and global clustering. An object \emph{Clustering} consists of many micro-clusters where each cluster is a cluster representative. Micro-cluster extends the cluster feature vector by adding the temporal components \cite{aggarwal2003framework}.

The first truly scalable clustering algorithm for data stream called BIRCH \cite{zhang1996birch} constructs a clustering structure in a single scan over the data with limited memory. The underlying concept behind BIRCH called the cluster feature vector is defined as follows
\begin{defn}\label{defn: cf}
Clustering Feature \cite{zhang1996birch} Given d-dimensional data
points in a cluster: $\left\{ \vec{X}\right\} $ where $i=1,2,..,N$,
the Clustering Feature (CF) vector of the cluster is a triple: $CF=\left(N,\, LS,\, SS\right)$
where $N$ is the number of data points in the cluster, $LS=\underset{i=0}{\overset{N-1}{\sum}}X_{i}$
is the linear sum of the data points in the cluster, and $SS=\underset{i=0}{\overset{N-1}{\sum}}X_{i}^{2}$
is the squared sum of the N data points. The cluster created by merging
two above disjoint clusters $CF_{1}$ and, $CF_{2}$ has the cluster feature is defined as
follows
\begin{equation}
CF=\left(N_{1}+N_{2},\, LS_{1}+LS_{2},\, SS_{1}+SS_{2}\right)
\end{equation}
\end{defn}
The advantages of this CF summary are: (1) it does not require to
store all the data points in the cluster; (2) it provide sufficient
information for computing all the measurements necessary for making
clustering decisions \cite{zhang1996birch}.
\begin{defn}\label{defn: tempcf} Micro-cluster \cite{aggarwal2003framework}. A micro-cluster for a set of $d-$dimensional points $X_{i_{1}},...,X_{i_{N}}$
with time stamps $T_{i_{1}},...,T_{i_{n}}$ is the $\left(2d+3\right)$-tuple
$\left(\overline{CF2^{x}},\overline{CF1^{x}},CF2^{t},CF1^{t},N\right)$,
wherein $\overline{CF2^{x}}$ and $\overline{CF1^{x}}$ each corresponds to a vector of $d$ entries. The definition of each of these entries is as follows
\begin{itemize}
\item For each dimension, the sum of the squares of the data values is maintained in $\overline{CF2^{x}}$. Thus, $\overline{CF2^{x}}$ contains $d$ values. The $p-th$ entry of $\overline{CF2^{x}}$  is equal to $\underset{j=1}{\overset{N}{\sum}}\left(X_{i_{j}}^{p}\right)^{2}$.
\item For each dimension, the sum of the data values is maintained in $\overline{CF1^{x}}$ . Thus, $\overline{CF1^{x}}$  contains $d$ values. The $p-th$ entry of $\overline{CF1^{x}}$  is equal to $\underset{j=1}{\overset{N}{\sum}}\left(X_{i_{j}}^{p}\right)$.
\item The sum of the squares of the time stamps $T_{i_{1}},...,T_{i_{n}}$ is maintained in $CF2^{t}$.
\item The sum of the time stamps $T_{i_{1}},...,T_{i_{N}}$ is maintained in $CF1^{t}$.
\item The number of data points is maintained in $N$.
\end{itemize}
\end{defn}

Micro clusters are mergable summaries which are useful for the problem of clustering streaming data as well as for the other problems \cite{Agarwal:2012:MS:2213556.2213562}. Because micro clusters can be merged into new one based on the additive property.

We note that Definition \ref{defn: tempcf} is the first definition of micro-cluster. In the later work, variants of micro-clusters are proposed to meet the specific requirements of the stream clustering algorithms. For example, in DenStream \cite{cao2006density}, micro-clusters fall into types such as core-micro-clusters, potential-c-micro-clusters, or outlier micro-clusters, which are used for density-based clustering.
In order to transmit a clustering structure over network, we convert it into a stream of binary bits, and restore the original clustering structure from the stream of bits by using two related techniques: serialization and deserialization (Figure \ref{fig:serialization}). The former converts an object into a stream of bits while the latter restores the original object from the streams of bits \cite{tanenbaum2009distributed}. By this technique, we can transmit the local clusterings and global clustering over the network. Therefore, the coordinator-site process can invoke the local clustering methods on the
data streams generated at the remote sites. We implemented clustering as a class \emph{Clustering} which implements the interface Serializable in Java, so that
we can transmit local clustering and global clustering over network as well as invoke a remote micro-clustering method on a remote site.
\begin{figure}
    \centering
         \includegraphics[width=0.45\textwidth]{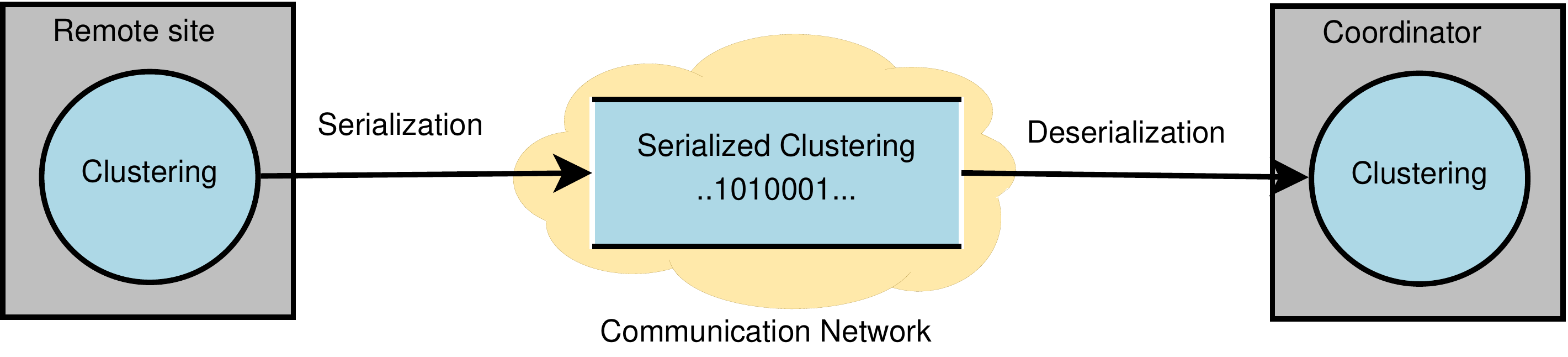}
    \caption{Serialization and deserialization of an object Clustering for transmitting it over network and restoring it}
    \label{fig:serialization}
\end{figure}

With the serialization technique, we can access to the attributes a well as invoke remote clustering methods on these local data streams as if they had lied in the same process.  As a result, we can use any macro-clustering algorithm on these micro-clusters or the variants of micro-clusters in order to create the global macro clustering.

\subsection{The Remote Process}
This subsection describes how a local micro-clustering algorithm at the remote site works. As we assumed in Section \ref{sec:formalmodel}, a remote site processes the incoming data from its local stream in an update epoch $t$: $S_{i}^{t}$.
For $i=1,..,N$ , the local clustering process for the push-based type of query at the remote site $i$ works as follows.
\begin{enumerate}
  \item make the connection between the remote site $i$ and the coordinator.
  \item while not at the end of the data stream $S_{i}$
  \begin{enumerate}
    \item initialize the stream learner $Mic\mathcal{A}(S_{i}^{t})$. 
    \item read block of data items at update epoch $t$
    \item create the local micro-clustering $LC_{i}^{t}$ by calling micro-clustering learner $Mic\mathcal{A}(S_{i}^{t})$ on block $S_{i}^{t}$.
    \item transmit the local micro-clustering to the coordinator.
  \end{enumerate}
\end{enumerate}
The clustering algorithm for push-based query, at each update epoch, every remote site generates and transmits the local micro-clustering to the coordinator.
For $i=1,..,N$ , the local clustering process for the pull-based type of query at the remote site $i$ works as follows.
\begin{enumerate}
  \item make the connection between the remote site $i$ and the coordinator.
  \item While no at the end of the data stream $S_{i}$
  \begin{enumerate}
    \item initialize the stream learner $Mic\mathcal{A}(S_{i}^{t})$. 
    \item read block of data items at update epoch $e$, let $S_{i}^{t}$ denote the block of data items.
  \end{enumerate}
\end{enumerate}
In contrast to the clustering algorithm for push-based type of query, in a clustering algorithm for pull-based query, the coordinator calls the local micro-clustering on the block of data items at epoch $t$ from the remote-site process.

\subsection{The Coordinator Process}
This subsection describes how the coordinator process works. Communication-efficiency of a distributed stream mining algorithm makes a big difference. One of the most efficient ways to improve communication efficiency is to move code to data instead of moving data to code \cite{Vassilvitskii2012}.

For each update epoch $t$, the coordinator works as follows.
\begin{enumerate}
  \item For each remote site $i$, perform the following steps
  \begin{enumerate}
    \item make the connection between the remote site $i$ and the coordinator.
    \item receive the local micro-clustering $LC_{i}^{t}$ from the remote site (in the case of push-based query), or remotely calls the local-micro clustering method from the remote site $i$ on the block $S_{i}^{t}$.
    \item add $LC_{i}^{t}$ to a buffer $buf$.
  \end{enumerate}
  \item If all the local micro-clusterings are received, the coordinator creates the global clustering $GC^{t}$ by calling $Mac\mathcal{A}$ on the micro-clusters in the buffer.
  \item process the global clustering $GC_{t}$.
\end{enumerate}
We note that step $3$ is executed on the basis of each specific context. For example, it may return the global clustering to the user as requested, or the global clustering can be stored in the history of global clusterings for some computation purpose, or may be sent back to all the remote sites.

\section{Theoretical Analysis}\label{theoretical}
In this section we present some theoretical results including the global clustering quality, communication complexity, and computational complexity.
\subsection{Global Clustering Quality}
Theorem \ref{thm1} asserts that the global clustering produced by $Dis\mathcal{A}$ in an update epoch is similar to the clustering created by the underlying algorithm $Dis\mathcal{A}$ on the same data set.
\begin{theorem} \label{thm1} Given a distributed algorithm for clustering streaming data $Dis\mathcal{A}$ which is based on the underlying two-phased stream clustering algorithm $\mathcal{A}$. The micro-clustering method $Mic\mathcal{A}$  at the remote sites create and maintain the micro-clusters, and the macro-clustering method $Mac\mathcal{A}$ creates and maintains the global macro clustering at the coordinator. Algorithm $Dis\mathcal{A}$ works in a network with the network topology of one coordinator and $N$ remote sites. Each remote site does not communicate with other remote sites directly but only with the coordinator. The global macro-clustering $GC^{t}$ which is created by the distributed framework $Dis\mathcal{A}$ at a any update epoch $t$ is similar to the clustering produced by the underlying centralized algorithm $\mathcal{A}$.
\end{theorem}

\begin{proof}

We first consider the case in which the clustering is generated by the centralized clustering algorithm $\mathcal{A}$. The underlying stream clustering algorithm consists of two methods: the online micro-clustering method $Mic\mathcal{A}$ and the macro-clustering method $Mac\mathcal{A}$.

Let $S_{i}^{t}$ be the streaming data block at the remote site $i$ in an update epoch $t$, for $i=1,..,N$, where $N$ is the number of remote sites.
We aim to prove that the global macro-clustering which is created by our framework $Dis\mathcal{A}$ is similar to the clustering which is generated by the centralized stream clustering algorithm $\mathcal{A}$.
We consider the $\mathcal{A}$ and $Dis\mathcal{A}$ in an update epoch. We also assume that all the data streams arrive at the remote sites with the same data speed, that means at each update epoch $t$, every remote site receives the same $M$ data items.
\begin{equation}
S^{t+\triangle t}=\underset{i=1}{\overset{N}{\cup}}S_{i}^{t}
\end{equation}
where $\triangle t$ is the time for the coordinator to receive all the streaming blocks from the remote sites.
Algorithm $\mathcal{A}$ works in two stages as follows
\begin{itemize}
\item \emph{Micro clustering:} the output of the micro-clustering method is given by
\begin{equation}\label{eq:centralmic}
LC_{centralized}=Mic\mathcal{A}\left(S^{t+\triangle t}\right)=\underset{i=1}{\overset{N}{\cup}}LC_{i}^{t}
\end{equation}
\item \emph{Macro clustering:} the output of the macro-clustering method is given by
\begin{equation}\label{eq:centralmac}
\mbox{\ensuremath{GC_{centralized}^{t}=Mac\mathcal{A}\left(\underset{i=1}{\overset{N}{\cup}}LC_{i}^{t}\right)}}
\end{equation}

\end{itemize}

We now consider how our distributed stream clustering framework $Dis\mathcal{A}$ works. As in the centralized stream clustering algorithm $\mathcal{A}$, the distributed framework $Dis\mathcal{A}$ also consists of two phases. However, a distinction between the centralized algorithm $\mathcal{A}$ and the distributed algorithm $Dis\mathcal{A}$ is that while the both micro-clustering method and macro-clustering method of $\mathcal{A}$ are executed in the same process,
the micro-clustering method of $Dis\mathcal{A}$ takes place at the remote sites, the macro-clustering method of $Dis\mathcal{A}$ occurs at the coordinator. The framework $Dis\mathcal{A}$ works as follows
\begin{itemize}
  \item \emph{At the remote site:} Each remote site generates the local macro-clustering $LC_{i}$, for $i=1,..,N$. If the clustering query is pushed-based query, remote sites send its local micro-clustering to the coordinator by the serialization technique.
  \item \emph{At the coordinator:} If the clustering query is pull-based query, the coordinator invokes remote clustering methods in order to get all the local micro-clusterings. The list of local clusterings received from remote sites
\begin{equation}\label{eq:dismic}
LC_{distributed}^{t+\triangle t_{1}}=\underset{i=1}{\overset{N}{\cup}}LC_{i}^{t}
\end{equation}
where $\triangle t_{1}$ is the time needed to transmit all the local clusterings
to the coordinator

The global clustering $GC^{t}_{distributed}$ is created by the macro clustering method $Mac$ up
to time $t+\triangle t_{1}+\triangle t_{2}$ is given by
\begin{equation}\label{eq:dismac}
GC^{t}_{distributed}=Mac\mathcal{A}\left(LC_{distributed}\right)=Mac\mathcal{A}\left(\underset{i=1}{\overset{N}{\cup}}LC_{i}^{t}\right)
\end{equation}
where $\triangle t_{2}$ is the time needed to produce global macro-clustering
(global clustering).
\end{itemize}

From the equations \ref{eq:centralmic}, \ref{eq:centralmac} and \ref{eq:dismic}, \ref{eq:dismac}, we can conclude that the global macro-clustering produced by the distributed framework $Dis\mathcal{A}$ in an update epoch $t$ is similar to the macro-clustering produced by the centralized version of
the two-phased stream clustering algorithm $\mathcal{A}$.
\end{proof}
\subsection{Communication Complexity}\label{subsec: comcomplex}
One of the fundamental issues of a distributed computing algorithm is to evaluate the communication complexity. The communication complexity of a distributed framework for clustering is the minimum communication cost  so that it produces the global clustering.

Let $b_{i}^{t}$ be the total bits sent between remote site $i$ and the coordinator. The size of a local micro-clustering $|LC_{i}|$ is defined as the product of the size of each micro-cluster $|CF|$  by the number of clusters $K_{i}$ in the local clustering structure  $|LC_{i}|=|CF|K_{i}$.
The number of bits used to transmit a local micro-clustering produced by the remote site $i$ to the coordinator site is $\log_{2}|LC_{i}^{t}|)$.
As such, the communication cost of all local clusterings from the remote sites to the coordinator is given by.
\begin{equation}
\underset{i=1}{\overset{N}{\sum}}\log_{2}|CF|K_{i}
=\underset{i=1}{\overset{N}{\sum}}\log_{2}K_{i}+\underset{i=1}{\overset{N}{\sum}}\log_{2}|CF|
\end{equation}

\begin{equation}
\underset{i=1}{\overset{N}{\sum}}\log_{2}|CF|K_{i}=\underset{i=0}{\overset{N}{\sum}}\log_{2}K_{i}+\frac{(N+1)N}{2}\log_{2}|CF|
\end{equation}

Theorem \ref{thm:thm1} let us know the communication cost needed to create a global clustering $GC^{t}$ in an update epoch $t$.
\begin{theorem} \label{thm:thm1}
Given a distributed algorithm for clustering distributed data streams $Dis\mathcal{A}$ which is based on the underlying two-phased stream clustering algorithm $\mathcal{A}$. The online phase $Mic\mathcal{A}$  at the remote sites create and maintain the micro-clusters, and the off-line phase $Mac\mathcal{A}$ creates the global macro clustering, the communication cost $\underset{i=1}{\overset{N}{\sum}}b_{i}^{t}$ which is needed to answer a clustering query at an update epoch $t$ is given by
\begin{equation}
\underset{i=1}{\overset{N}{\sum}}\log_{2}K_{i}+\frac{(N+1)N}{2}\log_{2}|CF|
\end{equation}
where $K_{i}$ is the number of micro-clusters in the local micro-clustering at the site $i$, $K_{GC}$ is the number of clusters in the global clustering $GC$, and $|CF|$ is the size of a micro-cluster.
\end{theorem}
Corollary \ref{lem:lem2} is a direct consequence of Theorem \ref{thm:thm1} when we apply Theorem \ref{thm:thm1} for the DisClustream which is based on the underlying clustering algorithm CluStream.
\begin{corollary}\label{lem:lem2} If the DistClustream is used, and the number of clusters at all the remote sites and the coordinator are the same, the communication cost $\underset{i=1}{\overset{N}{\sum}}b_{i}^{t}$ which is needed to answer a clustering query at an update epoch $t$ is given by
\begin{equation}
N \log_{2}K + \frac{(N+1)N}{2}\log_{2}|CF|
\end{equation}
where $K$ is the number of micro-clusters in a local micro-clustering, $K_{GC}$ is the number of clusters in the global clustering $GC$, and $|CF|$ is the size of a micro-cluster.
\end{corollary}

\begin{figure}
    \centering
         \includegraphics[width=0.45\textwidth]{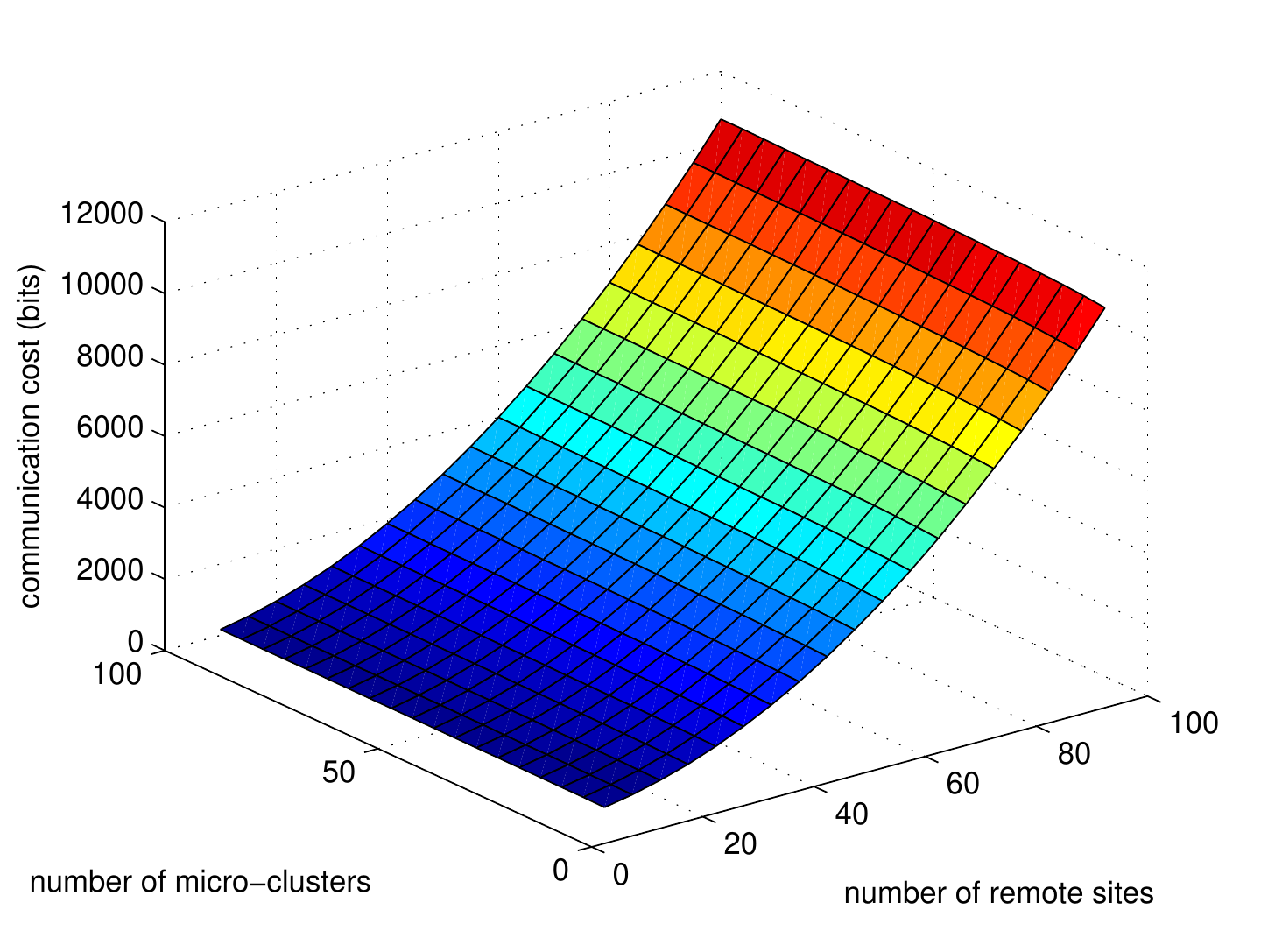}
    \caption{The communication cost needed to generate the global clustering in DisClustream in an update epoch}
    \label{fig:communicationcost}
\end{figure}

\subsection{Computational Complexity}
This section answers the question how much time is needed to compute the global clustering. Let $\mathcal{T}_{mic}$ and $\mathcal{T}_{mac}$ denote the time needed to produce a micro-clustering and  a macro-clustering respectively. The time needed to generate the clustering by an underlying two-phased stream clustering $\mathcal{T}_{centralized}$ is given by
\begin{equation}\label{eq:centralizedtime}
\mathcal{T}_{centralized}=N \mathcal{T}_{mic} + \mathcal{T}_{mac}
\end{equation}

The time needed to generate the global clustering is computed from the time at which all the local micro-clustering algorithms start until the time at which the global macro-clustering is generated. Therefore, the time needed to compute the global clustering $\mathcal{T}_{distributed}$ is given by
\begin{equation}\label{eq:distributedgctime}
\mathcal{T}_{distributed}=\mathcal{T}_{mic} + \mathcal{T}_{transmit}^{all LC}+ \mathcal{T}_{mac}
\end{equation}

Let $\mathcal{T}_{transmit}^{LC_{i}}$ denote the time needed for transmitting the local clustering from the remote site $i$ to the coordinator. As our framework is multi-threading organized, the total time needed to transmit all the local clusterings $\mathcal{T}_{transmit}^{all LC}$ is less than or equal to the sum of time needed to transmit each local clustering $\mathcal{T}_{transmit}^{LC_{i}}$, that means $\mathcal{T}_{transmit}^{all\, LC}\leq \underset{i=1}{\overset{N}{\sum}}\mathcal{\mathcal{T}}_{transmit}^{LC_{i}}$.

From the equation \ref{eq:distributedgctime}, we can reduce the time to produce the global clustering at the coordinator, the following tasks should be done: (1) speed up the local clustering algorithm (reduce $\mathcal{T}_{clust}^{LC}$); (2) reduce the time  to send the local clustering $\mathcal{T}_{send}^{LC}$; (3) speed up the cluster ensemble algorithm at the coordinator site. The first task depends on the selection of stream clustering algorithm at remote sites. For the second task we reduce the size of local clusters transmitted. The third task depends on the selection of the ensemble cluster approach. The time to send a local clustering to the coordinator site depends on the network bandwidth, and the size of local clustering. As shown in the experiment part, the local micro-cluster is small.
The speedup of our framework is given by
\begin{equation}\label{eq:speedup}
speedup=\frac{\mathcal{T}_{centralized}}{\mathcal{T}_{distributed}}
\end{equation}
where $\mathcal{T}^{centralized}$ is the time needed to generate the clustering by the underlying two-phased stream clustering algorithm and $\mathcal{T}_{distributed}$ is the time needed to generate the global clustering by our framework.

From the equtions \ref{eq:centralizedtime}, \ref{eq:distributedgctime}, and \ref{eq:speedup}, we have
\begin{equation}
speedup=\frac{N \mathcal{T}_{mic} + \mathcal{T}_{mac}}{\mathcal{T}_{mic} + \mathcal{T}_{transmit}^{all LC}+ \mathcal{T}_{mac}}
\end{equation}
\begin{figure}
    \centering
         \includegraphics[width=0.45\textwidth]{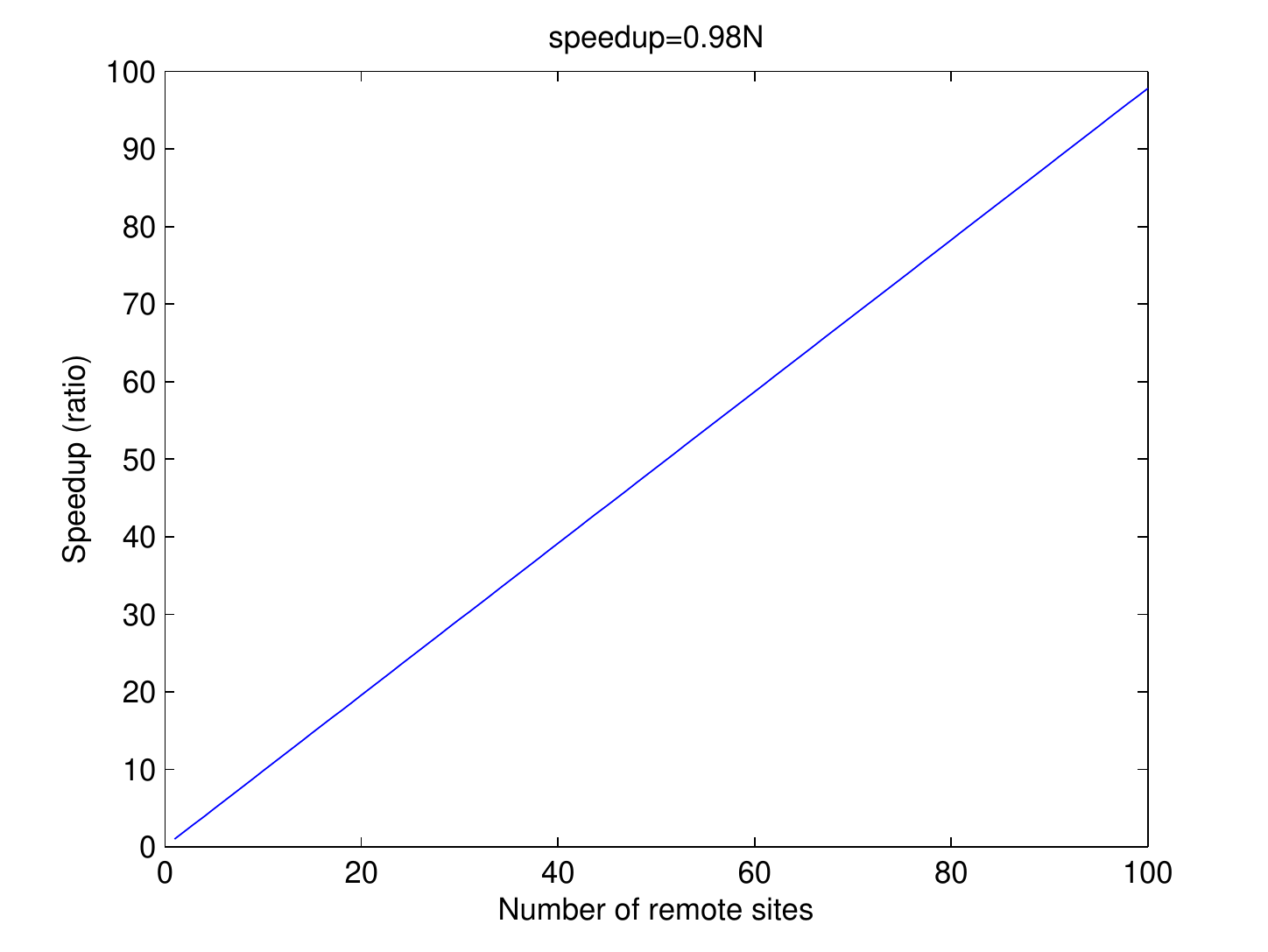}
    \caption{Speed up scales with the number of remote sites}
    \label{fig:speedup}
\end{figure}
Based on the experiment in Subsection \ref{scaleinkernels}, the compute the average values of $\mathcal{T}_{mic}$, $\mathcal{T}_{mac}$, and $\mathcal{T}_{transmit}^{LC_{i}}$. From these parameters, The speed up function was approximately given by $speedup=0.98N$. Figure \ref{fig:speedup} shows that, the speedup scales with the number of remote sites. In conclusion, the speed of our framework scales up with the increasing number of remote sites. In other word, our framework scale in the size of network.
\section{Empirical Results}\label{sec:Evaluation} \label{sec:eval}
All the experiments were performed on a Intel Pentium (R) $2.00GHz$ computer with $1.00GB$ memory, running Windows XP professional.

We implemented the proposed framework in client/server model in Java. We used the two-phased stream algorithms which were implemented in  \href{http://moa.cs.waikato.ac.nz}{the MOA package} \cite{bifet2010moa}.
We implemented an instance of the proposed framework, which we call DisClustream (Distributed Clustream). DisClustream was developed on the basis of the underlying stream clustering algorithm Clustream.
The streaming data sets (Sensor Stream, Powersupply Stream, and Kddcup99) that were used in our experiments were from the Stream Data Mining Repository \cite{zhu2010}. Forest Covertype data set was downloaded from \href{http://moa.cs.waikato.ac.nz}{the page of MOA}\cite{bifet2010moa}. Depending on the purpose of each group of experiments, we selected some data sets of the above data sets.

We sought to answer two empirical questions about our framework: (1) How accurate our clustering framework for distributed data streams in comparison with the central clustering approach on the same data set is? (2) How scalable our framework is?
All the empirical evaluations were done in an update epoch.
\subsection{Evaluation on Global Clustering}
This group of experiments evaluated aspects of the global clustering in terms of clustering quality compared with the centralized clustering algorithm, time needed to create the global clustering.

\subsubsection{Comparison of Clusterings produced by $\mathcal{A}$ and $Dis\mathcal{A}$}
The goal of these experiments is to compare the clusterings produced by centralized and distributed stream clustering algorithms. We expect that the global clustering produced by distributed stream clustering will be almost the same as the clustering produced by centralized stream clustering algorithm.
For ease of comparison, we chose a data set in which the number of attributes and the number of classes are small sufficient to visualize clusterings. The appropriate data set is  for this task is Powersupply used to predict which hour the current supply belongs to. This data set includes $29928$ records, each record consists of $2$ attributes. These records can fall into one of $24$ classes. To obtain clusterings, we ran the centralized version of ClusStream, and the distributed version DisCluStream. For the centralized CluStream, the entire data set can be seen as the incoming data stream. For the distributed DisCluStream, we divided the data set into two data sets of the same size. There were two remote sites, each remote site produced a local micro clustering from one of two above data sets.

Figure \ref{fig:centraldistcompare} illustrates two clusterings produced by centralized (marked by square) and distributed (marked by circle) stream clustering algorithms.
\begin{figure}
    \centering
         \includegraphics[width=0.45\textwidth]{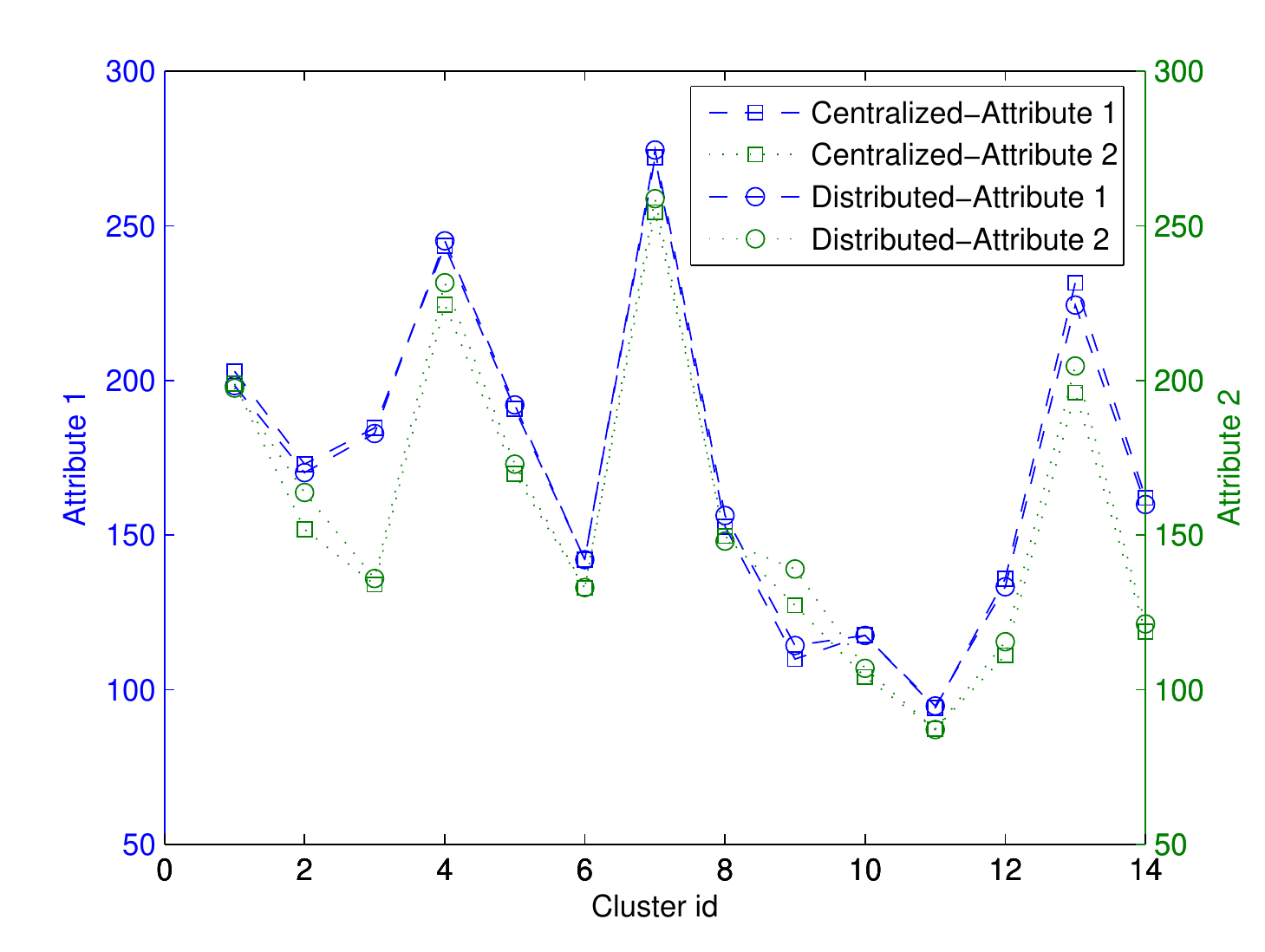}
    \caption{A comparison of two clusterings which were created by distributed stream clustering algorithm and the centralized one on Power Supply data set}
    \label{fig:centraldistcompare}
\end{figure}
As we expect, the clustering produced by distributed clustering algorithm DisClustream was almost the same as the clustering produced by centralized one CluStream (Figure \ref{fig:centraldistcompare}). We obtained the same results when testing with the data sets kddcup99, covertype. However, due to the exposition, in this work, we only illustrated with the data set Powersupply.
We also experimented with the data sets KDDCup99, Covertype data sets, but for the purpose of illustration, in this section, we chose a data set in which the number of attributes and the number of classes are small sufficient to visualize clusterings.
In conclusion, the global clustering produced by our distributed clustering framwork is almost similar to the clustering produced by centralized one.
The quality of an underlying local stream clustering affects the overall quality of the global clustering. As such, the core issue is still how to select existing stream clustering algorithms, or develop fast and high-quality stream clustering algorithms from which we can extend to our proposed framework.
\subsubsection{Time to generate global clustering}
We should distinguish between the time needed to produce the global clustering and the time needed to run the macro-clustering algorithm. The time needed to generate the global clustering in the distributed framework is the duration from the time at which all the local micro-clustering algorithms start until the time at which the global macro-clustering is generated. The time needed to generate micro clusters at the remote site is much greater than the time needed to run macro clustering algorithm in order to generate the global clustering at the coordinator. This is a direct consequence of the two-phased stream clustering approach in which the online phase consumes more time than the off-line one \cite{aggarwal2003framework}. Our the experiment with KDDCup99 data set shows that,time to produce micro clusters at the remote site (14021 instances, and number of micro-clusters 100, time needed is 25701 ms)  is much greater than the time needed to produce macro-clustering (391 ms) at the coordinator. Therefore, in order to speed up the algorithm, we should select, or develop the fast and high-quality micro-clustering approach.

\subsection{Evaluation on Communication Efficiency}
The chief purpose of this subsection is to evaluate the time needed to transmit block of streaming data with the time needed to transmit the local micro-clustering built from the corresponding the block of streaming data. As shown in Section \ref{subsec: comcomplex}, the communication overhead needed to transmit local micro clusterings would be negligible compared to the raw data transmission involved. In other word, the time needed to transmit local micro-clustering would be much smaller than the time needed to transmit the raw streaming data.

We execute this groups of experiments with DisClustream on the Intel data set. We divided this data set into many data sets based on the sensor id. Without loss of generality, the sensor data set from sensor $1$ to $10$ was used as data streams at the remote sites. The number of micro-clusters was set to default number ($100$ kernels, in Clustream). For simplicity, we set up $10$ experiments with $10$ sensor data sets from $1$ to $10$ separately, that means each experiment include one coordinator and one remote site. Figure \ref{fig:sendtime} shows that, the time needed to transmit raw streaming data depends on the size of streaming data block. The time needed to transmit micro-clusters is almost invariant as the number of micro-clusters were set to the same number $100$ in all experiments.

Our experiments shows that the micro-clusters produced by remote sites are small sufficient for meeting the requirement of communication efficiency.

\begin{figure}
    \centering
         \includegraphics[width=0.45\textwidth]{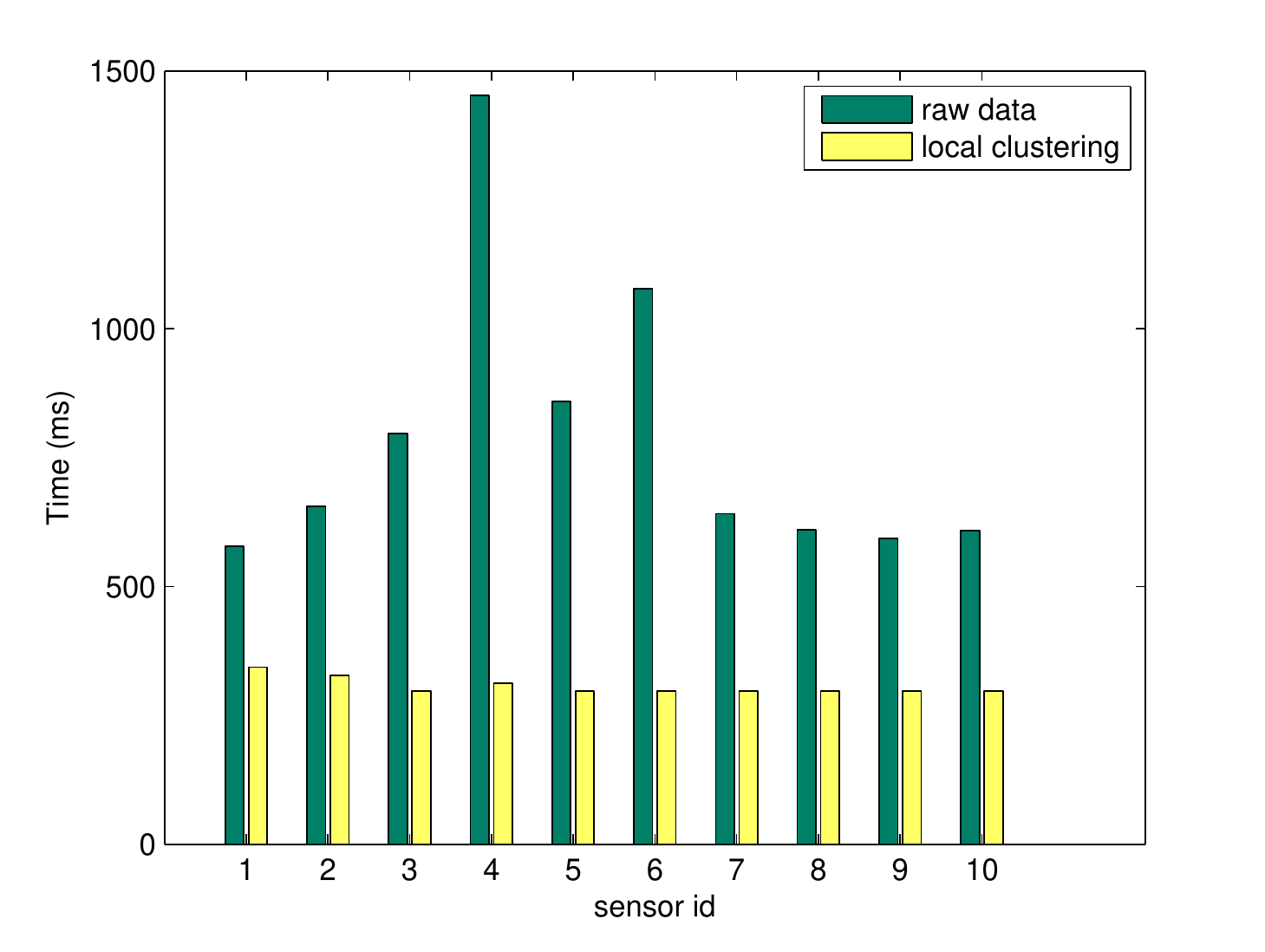}
    \caption{Time to transmit block of streaming data vs. time to transmit local-micro clustering}
    \label{fig:sendtime}
\end{figure}
The goal of this experiment is to determine the size of local clustering. The DisCluStream was used in this experiment. We used KDD Cup99 for this experiment. The framework consists of one remote site and one coordinator. At the remote site, the number of micro-clusters was fixed to 100. The coordinator received the local clustering and wrote it to file. We changed the number of instances in the window (an update epoch) in the range 1000, 2000, 3000, 4000, 5000. The size of local clustering file was invariant and equals to 13.2KB. In fact, the actual size of local clustering in memory may be smaller than 13.2KB as the local clustering file includes the size of file format. The size of local clustering is invariant because we fixed the number of micro-clusters to $100$. Therefore, the size of local clustering in DisClustream is independent of the number of instances that a remote site receives in an update epoch.

\subsection{Evaluation on Scalability}
The success of a distributed framework for clustering streaming data depends on the scalability of the local stream clustering algorithm and the scalability in term of number of remote sites.
The scalability of the micro-clustering algorithm used at the remote sites is mandatory because it must process the large volume of incoming data. The scalability of the micro-clustering algorithm Clustream in the number of data dimensions, and the number of clusters  was thoroughly evaluated \cite{aggarwal2003framework} in term of the time needed to produce clustering. In contrast, we evaluated the scalability of DisClustream in term of communication time by increasing the number of micro-clusters, and the number of instances in the window.
\subsubsection{Scalability in Window Size}
This group of experiments studied the scalability in the increasing number of instance in the window. The underlying stream clustering algorithm was used to evaluate the scalability of our proposed framework is Clustream. As such, the micro-clustering algorithm at the remote sites the micro-clustering method of Clustream while the macro-clustering algorithm at the coordinator is K-means.
Figure \ref{fig:scale} shows how our distributed stream clustering framework scale with the number of instances in an update epoch. As  we experimented with the sensor data set in which each record consisted of two attributes, we can observe the time to transmit the local clustering to the coordinator in Figure \ref{fig:scale} .
We also tested the scalability in window size with KDD Cup data set. The result of this experiment demonstrated that, the time needed to create local clustering scales with the number of instances in an update epoch. However, the time needed to transmit the local clustering is hardly varied. For example, with KDD Cup data set, the time needed to create the local clustering (18,11 seconds) was much larger than the time needed to transmit the local clustering (46 miliseconds) with the number of instances in an update epoch 10000. As the number of instances in an update epoch increases 10 times (100000), the time needed to create local clustering was almost 3 (minutes) while the time needed to transmit local clustering was only 47 (milliseconds)

\begin{figure}
    \centering
         \includegraphics[width=0.45\textwidth]{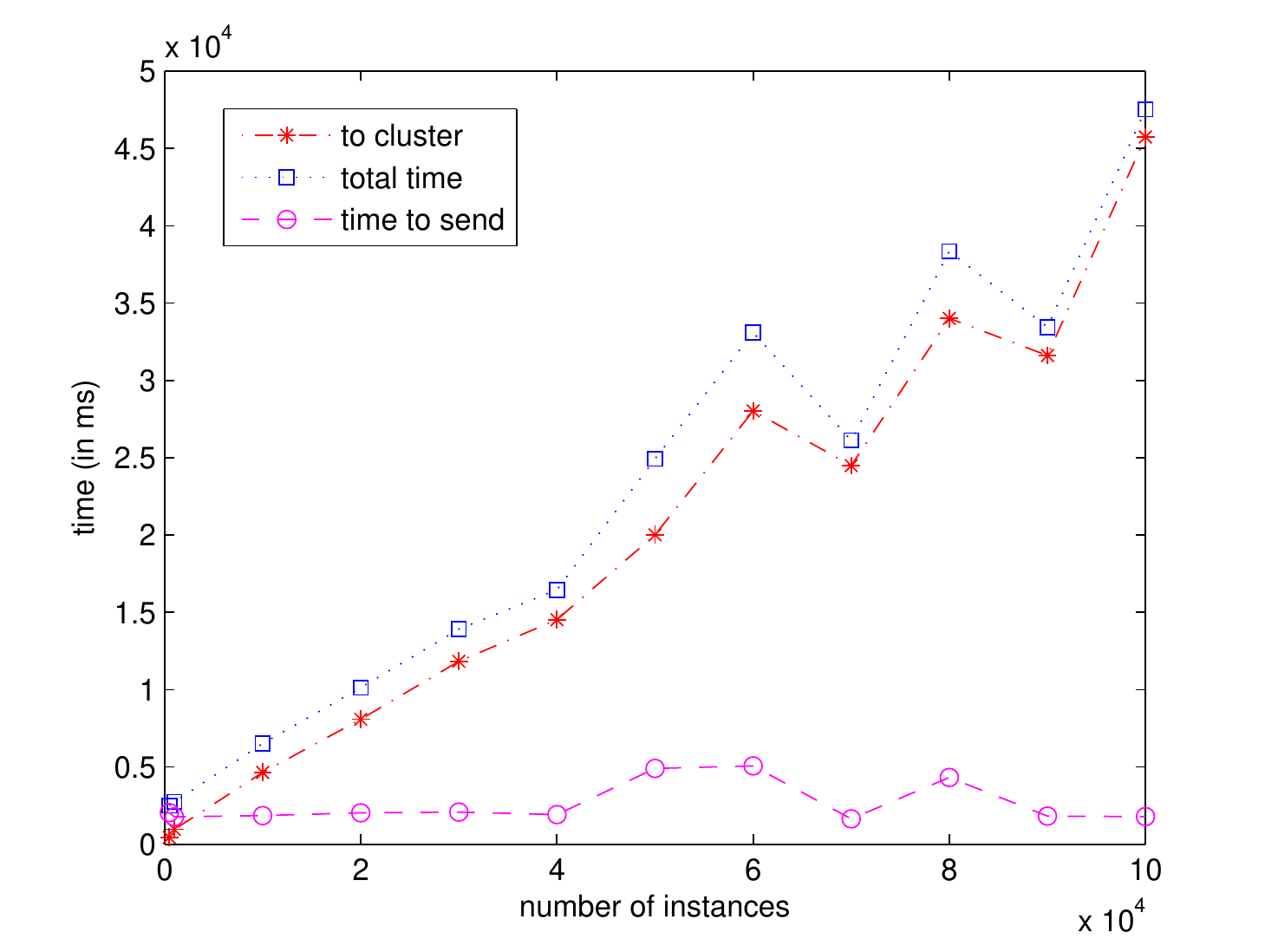}
    \caption{Scalability in window width}
    \label{fig:scale}
\end{figure}

\subsubsection{Scalability in Number of Micro clusters}\label{scaleinkernels}
To evaluate the scalability of our algorithm in the increasing number of micro clusters. Evaluation on the scalability of our algorithm was done on two data sets  Intel sensor data set and KDDCup 99 data set. We set up experiment as follows. The number of instances was fixed to 10000 instances. There were nine remote sites connected to the coordinator sites. The number of micro clusters was selected from the range $(50,100, 150,250,350,400,450,500)$ as shown in Figure \ref{fig:scaleinkernels}. An increase in the number of micro-clusters resulted in the increasing time to transmit micro-clusters 
\begin{figure}
    \centering
         \includegraphics[width=0.45\textwidth]{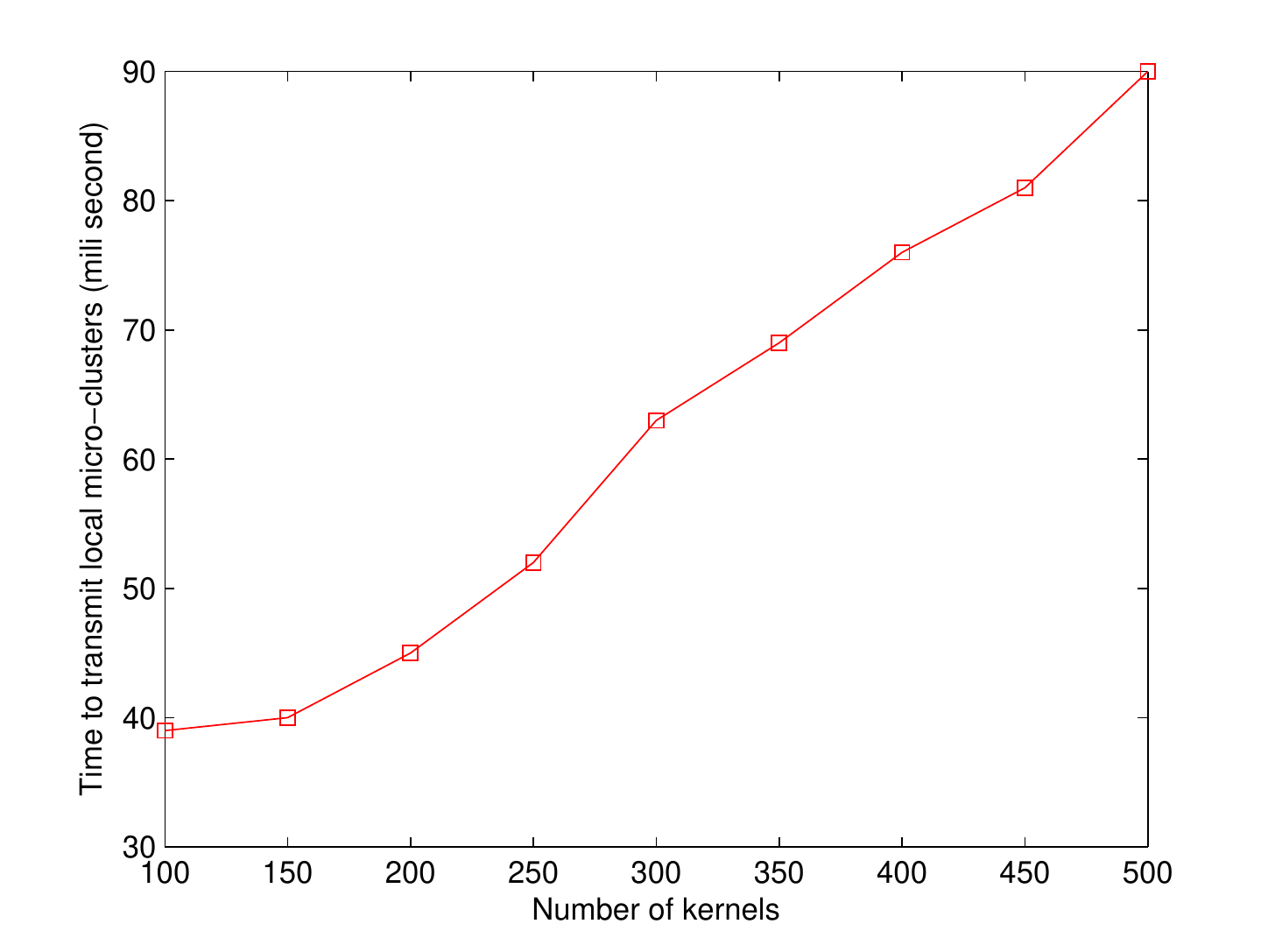}
    \caption{Scalability in number of micro-clusters}
    \label{fig:scaleinkernels}
\end{figure}
In our experiments, time required to produce micro-clusters is much larger than the time required to transmit local micro-clusters.


In conclusion, our experiments on the scalability with DisClustream demonstrated that, the time needed to create local clustering in an update epoch scales with the increasing number of instances in the window, and the increasing number of micro-clusters while the time needed to transmit a local clustering is almost invariant. The data stream clustering algorithms that are fast sufficient to deploy them in resource-limited applications such as sensor networks are still open questions.
\section{Further Discussion}
Micro-clustering approach is used in many tasks of data mining for instance, in \cite{masud2008practical}, micro-clustering method is used to summarize cluster information for the data stream classification.


Our theoretical and empirical results shows that, we can develop a general framework for mining distributed streaming data by using micro-clusters. The remote sites create and maintain micro-clusters or the variants of micro-clusters. The coordinator can generate and maintain the global pattern from the micro-clusters which are received from the remote sites by using the micro-clusters based data mining such as classification, clustering, and frequent patterns. The reasonable foundation behind this local micro-clustering approach to data stream mining is that micro-clusters maintain statistics at a sufficiently high level of granularity (temporal and spatial)\cite{aggarwal2003framework}. As such, micro-clustering algorithms can be seen as a data summarization method.
For the problem of clustering distributed streaming data, there are some open questions such as: How to create and maintain the global clustering by cluster ensembles method on the micro-clusters which are received from the remote sites? How to develop the incremental micro-clustering algorithms for clustering distributed streaming data instead of the periodic approach?
\section{Related Work}\label{sec:Related-Work}
Stream clustering algorithms basically fall into two categories: single-phased
and two-phased approaches. Single-phased clustering algorithm first
divides the data stream into finite segments. Clustering data stream
then reduces to the clustering of finite segments by using K-means
algorithm \cite{guha2000clustering,o2002streaming}. A big drawback of this approach is that it is incapable
of monitoring the evolving data streams.

The first two-phased stream clustering algorithm Clustream introduced by Aggarwal et al \cite{aggarwal2003framework}
has two separate parts: online and off-line components. The role of
the online component is to process and to extract summary information
from data stream while the off-line builds the meaningful clustering structure
from the extracted summary information.

Clustream can efficiently maintain a very large number of micro-clusters,
which are then used to build macro-clusters. Further, clustering structure
can be observed in arbitrary time horizons by using pyramidal time
model. Hence, Clustream is capable of monitoring the evolution of
clustering structure. However, one of the main drawbacks of Clustream
is that it requires the number of clusters in advance. To overcome this drawback, Cao et al. \cite{cao2006density} have proposed a density-based clustering data streams. This approach has an advantage is that it can handle outliers in data streams, and adapt to the
changes of clusters. However, the major limitation of DenStream is
that it frequently runs the off-line component in order to detect the
changes of clusters. As such, the off-line component consumes the cost
of computing clusters much.  In the experimental evaluation, we use two well-known algorithms Clustream \cite{aggarwal2003framework}, and DenStream \cite{cao2006density} as the local stream clustering algorithms at the remote sites.

Cormode et al. \cite{cormode2007conquering} presented a suite of k-center-based algorithms for clustering distributed data streams. The communication topology consists of $N$ remote sites, and one coordinator. Two any remote sites can directly communicate with each other. Their approach can provide a global clustering that is as good as the best centralized clusterings. However, this approach provide a approximate clustering by using the underlying k-centers stream clustering algorithm.

Zhou et al. \cite{zhou2007distributed} considered the problem of clustering distributed data stream by using EM-based approach. The advantage of this approach is that it can deal with the noisy and incomplete data streams. To deal with the evolving data streams, they used the reactive approach to rebuild the local clustering and global clustering when it no longer suits the data. However, it is not scalable in the large number of remote sites such as in a sensor network of thousands of nodes. Although this drawback can be overcome by model-merging technique, the global clustering quality considerably reduces. Furthermore, the memory consumption at the remote site increases with the increasing number of data dimensionality and the increasing number of models.

In contrast to the work introduced by Zhou et al. \cite{zhou2007distributed}, Zhang et al.\cite{Zhang_08} introduced a suite of k-medians-based algorithms for clustering distributed data streams, which work on the more general topologies: topology-oblivious algorithm, height-aware algorithm, and path-aware algorithm. To deal with the evolving data streams, they selected periodic approach in which the global clustering at the root site is continuously updated every period called update epoch. Similar to , the global clustering is approximately computed on the summaries that are received from the internal and leaf nodes.

\section{Conclusions}\label{sec:Summary}
We have proposed a distributed framework for clustering streaming data by extending the two-phased stream clustering approach which is widely used to cluster a single data stream. While the remote sites create and maintain micro-clusters by using the micro-clustering method, the coordinator create and maintain the global clustering by using the macro-clustering algorithm. Remote sites send the local micro-clusterings
to the coordinator by the serialization technique, or the coordinator invokes the remote methods in order to get the local micro-clusterings from the remote sites.
We have theoretically analyzed our framework in the following aspects. The global clustering which is created in an update epoch is similar to the clustering which is created by the centralized stream clustering algorithm. We also estimated the communication cost as well as proved that the speed of our framework scales with the number of remote sites.

Our empirical results demonstrated that, the global clustering
generated by our distributed framework was almost the same as the
clustering generated by the underlying centralized algorithm on
the same data set with low communication cost.

In conclusion, by using the local micro-clustering approach, our framework achieves scalability, and communication-efficiency while assuring the global clustering quality. Our result which was presented in this paper along with other work can be seen as one of the firs steps towards a novel distributed framework for mining streaming data by using micro-clustering method as an efficient and exact method of data summarization.

 \bibliographystyle{plain}
 \bibliography{sigmod2012}

\end{document}